\documentclass[11pt]{article}
\usepackage{graphicx}
\usepackage{subfigure}
\usepackage{tikz,tikz-network}
\usepackage{float}
\usetikzlibrary{positioning,petri}
\usepackage{amsfonts,amsbsy,amssymb,amsmath,amsthm,amsfonts,mathtools}
\usepackage{verbatim}
\usepackage{parskip}
\usepackage[T1]{fontenc}
\usepackage{multicol}
\usepackage[onehalfspacing]{setspace}
\usepackage{color}
\usepackage[autostyle]{csquotes}
\usepackage{pgf}
\usepackage{hyperref}
\usepackage{thmtools}
\usepackage{caption}
\usepackage{soul}

\declaretheoremstyle[%
  spaceabove=-6pt,%
  spacebelow=6pt,%
  headfont=\normalfont\itshape,%
  postheadspace=1em,%
  qed=\qedsymbol%
]{mystyle} 

\hypersetup{
    colorlinks=true,
    linkcolor=blue,
    filecolor=blue,      
    urlcolor=blue,
    citecolor=blue,
    pdftitle={Overleaf Example},
    pdfpagemode=FullScreen,
    }

\newtheorem{thm}{Theorem}[section]
\newtheorem{prop}{Proposition}[section]
\newtheorem{cor}[thm]{Corollary}

\newtheorem{lemma}[thm]{Lemma}

\newtheorem{example}[thm]{Example}

\newtheorem{construction}[thm]{Construction}

\newcommand{\FF}{{\rm I\!F}}

\newcommand{\CC}{\mathcal{C}}

\newcommand{\ba}{\mathbf{a}}

\newcommand{\wt}{\mathtt{wt}}
\newcommand{\hull}{\mathtt{hull}}
\title{On Construction of  Linear   (Euclidean) Hull Codes over Finite Extensions  Binary Fields}
\author{Sanjit Bhowmick$^{1}$, Deepak Kumar Dalai$^{1}$, Sihem Mesnager$^{2}$
\footnote{
$^{1}$School of Mathematical Sciences,
National Institute of Science Education and Research,\\
An OCC of Homi Bhabha National Institute, Bhubaneswar, Odisha 752050, India. Email: sanjitbhowmick392@gmail.com; deepak@niser.ac.in\\
$^{2}$Department of Mathematics, University of Paris VIII, F-93526 Saint-Denis, Laboratory Analysis, Geometry and Applications, LAGA, University Sorbonne Paris Nord, CNRS, UMR 7539, F-93430, Villetaneuse, France and Telecom Paris, Polytechnic institute of Paris, 91120 Palaiseau, France. Email: smesnager@univ-paris8.fr 
}}

\begin{document}
\maketitle
\begin{abstract} 

The hull of a linear code is defined as the intersection of the code and its dual. This concept was initially introduced to classify finite projective planes. The hull plays a crucial role in determining the complexity of algorithms used to check the permutation equivalence of two linear codes and compute a linear code's automorphism group. Research has shown that these algorithms are very effective when the hull size is small. Linear complementary dual (LCD) codes have the smallest hulls, while codes with a one-dimensional hull have the second smallest. 

A recent notable paper that directs our investigation is authored by H. Chen, titled ``On the Hull-Variation Problem of Equivalent Linear Codes", published in IEEE Transactions on Information Theory, volume 69, issue 5,  in 2023. In this paper, we first explore the one-dimensional hull of a linear code over finite fields. Additionally, we demonstrate that any LCD code over an extended binary field \( \FF_q \) (where \( q > 3 \)) with a minimum distance of at least $2$ is equivalent to the one-dimensional hull of a linear code under a specific weak condition. Furthermore, we provide a construction for creating hulls with \( \ell + 1 \)-dimensionality from an \( \ell \)-dimensional hull of a linear code, again under a weak condition. This corresponds to a particularly challenging direction, as creating \( \ell \)-dimensional hulls from \( \ell + 1 \)-dimensional hulls.  Finally, we derive several constructions for the \( \ell \)-dimensional hulls of linear codes as a consequence of our results. 
\end{abstract}

\noindent\textbf{Keywords:}
Linear code, Hull of a linear code, Dual code, LCD code.\\
 \noindent\textbf{2020 AMS Classification Code:} 94B05; 94B27. \\

%\end{keyword}

%\end{frontmatter}
%Highlights of this paper:

%\begin{itemize}
 %   \item Compute the RREF of any bimatrix over $\FF_qR.$
  %  \item Determine a generator and party-check bimatrix for any $R$-linear code over $\FF_qR$.For submission, it
   % \item Survey the Galois invariance for $R$-linear codes over $\FF_qR$.
%\end{itemize}

\section{Introduction}\label{sec:intr}

For a linear code \(\CC := [n,k]\), the hull of \(\CC\) is defined as \(\hull(\CC) = \CC \cap \CC^\perp\). Assmus and Key first introduced this concept in \cite{AK1990}. Specifically, if \(\dim(\hull(\CC)) = \ell\), we say that \(\CC\) has an \(\ell\)-dimensional hull. In particular, if \(\ell = 0\), \(\CC\) is an LCD code; if \(\ell = k\), \(\CC\) is self-orthogonal, meaning that \(\CC \subseteq \CC^\perp\). In the case where \(\ell = n-k\), it follows that \(\CC^\perp \subseteq \CC\). Therefore, we can conclude that for a linear code \(\CC := [n,k]\) with \(\dim(\hull(\CC)) = \ell\), the dimension \(\ell\) satisfies \(0 \leq \ell \leq \min\{k,n-k\}\).

From the coding theory perspective, the hulls of linear codes play a significant role in determining the complexity of algorithms used to verify the equivalence of permutations of two linear codes. It has been demonstrated that a smaller hull size indicates greater complexity in the algorithms for calculating the group of automorphisms of a linear code \cite{car18, LZ2019}. In particular, codes with low-dimensional hulls have been extensively studied. Many researchers have made significant progress in exploring the conditions necessary and sufficient for a linear code to have a one-dimensional hull (see \cite{LZ2019}) and in designing one-dimensional hull linear codes, such as one-dimensional cyclic codes (see \cite{LZ2019}) or generalizations of Reed-Solomon (RS) codes that are maximum distance separable (MDS).

One case of low hull codes is the $0$-hull codes, also known as Linear Complementary Dual (LCD) codes. It has been demonstrated by Carlet and Guilley (\cite{CG16}) that binary LCD codes are appealing for their cryptographic applications against side-channel and fault injection attacks. Since then, considerable work has been done on LCD codes in recent years to advance this direction, notably for a better understanding of these objects. A main crucial classification result was given by the authors of  \cite{CMT18} on Linear codes over $\FF_{q}$, which are equivalent to LCD codes for $q>3$. Therefore, the $q$-ary Euclidean LCD codes are as good as $q$-ary linear codes for $q > 3$. Further, Carlet, Li, and  Mesnager (\cite{car18}) presented character sums in semi-primitive cases for constructing LCD and linear codes with a one-dimensional hull from cyclotomic fields and multiplicative subgroups of finite fields. For several years, the remaining investigation has focused on binary or ternary cases. Since then, many results have been obtained on the remaining LCD code cases and their classification over finite and other algebraic structures. Notably,  Anderson et al.  (\cite{Anderson-et-al})
presented a simple proof of the equivalence of $\FF_{q}$-linear codes with relative (Euclidean)Hulls and also  Luo et al. presented in \cite{Luo-et-al}  a simple proof of the equivalence of $\FF_{q^2}$-linear codes with diverse Hermitian hull dimensions for $q > 2$ (see Theorem 5), extending the main classification results, presented for trivial hull in \cite{CMT18} (which was focused on LCD codes)  to any dimensional hull codes.

In the study by Chen \cite{Chen23}, an important issue regarding the hull-variation problem of equivalent codes is examined, with implications for coding theory and cryptography. Chen introduced the concept of the maximal dimensional hull as an invariant related to a linear code and equivalent transformations. It was demonstrated that for any integer \(h\) such that \(0 \leq h \leq (n-1)\), a \([2n,n]_q\) self-dual linear code is equivalent to a linear \(h\)-dimensional hull code. In contrast, a linear LCD code over \(\FF_{2^s}\), with minimum distance \(d \geq 2\) and dual distance \(d^\perp \geq 2\), is equivalent to a linear one-dimensional hull code under specific weak conditions. Furthermore, it was shown that for any fixed \(h \leq 0\), there exists an asymptotically good sequence of linear \(h\)-dimensional hull codes over \(\FF_{q}\) for \(q \leq 64\) and \(q \neq 125\), which exceeds the Gilbert-Varshamov bound. Several new families of LCD negacyclic codes over \(\FF_{q}\) and LCD cyclic codes over \(\FF_{3}\) have also been constructed. The method can be applied to build arbitrary-dimensional Hull MDS codes from generalised Reed-Solomon codes and generalised twisted Reed-Solomon codes over \(\FF_{2^s}\).

Recently, Li, Shi, and Kim \cite{LSK23} investigated one-dimensional hulls for binary linear codes. In their paper, they constructed a one-dimensional hull for binary linear codes derived from a binary LCD code. They also produced optimal binary linear codes with one-dimensional hulls. Several other studies (see \cite{Sok20, Sok22, Sok22-bis, Sok22-ter}) focus on the investigation or construction of one-dimensional hulls and higher-dimensional hulls. Much of this work emphasizes deriving \((\ell-1)\)-dimensional hulls from \(\ell\)-dimensional hulls. Additionally, some research aims to explore the opposite direction, obtaining a one-dimensional hull of non-binary codes from LCD codes (i.e., \(0\)-dimensional hulls). 

Currently, to our best knowledge, no articles address the construction of \((\ell+1)\)-dimensional hulls of non-binary codes from \(\ell\)-dimensional hulls of non-binary codes, representing the inverse of studying \((\ell-1)\)-dimensional hulls from \(\ell\)-dimensional hulls. This paper establishes a connection between \((\ell+1)\)-dimensional hulls of non-binary codes and \(\ell\)-dimensional hulls of non-binary codes.

In this paper, we first focus on recent remarkable results by Chen (\cite{Chen23}), who demonstrated, notably,  that a linear complementary dual (LCD) code over extended binary fields, with a minimum distance \(d \geq 2\) and a dual distance \(d^\perp \geq 2\), is equivalent to a linear code with a one-dimensional hull under certain weak conditions. We show that any linear code over an extended binary field, with a minimum distance of at least 2, is equivalent to a linear code with a one-dimensional hull when a weak condition is met. Furthermore, expanding on this result, we demonstrate that any \(\ell\)-dimensional hull of a linear code over the extended binary field, with a minimum distance of at least 2, is equivalent to a linear code with a \((\ell + 1)\)-dimensional hull under a weak condition, for \(0 < \ell < \min\{k, n-k\}\).

We use the term  `` weak condition '' in Chen's paper (\cite{Chen23}) to describe a condition easily verified for a given class of codes. Several classes of linear codes satisfying this condition can be readily identified. In other words, it is not difficult to construct codes that meet the specified conditions.

The paper is organized as follows. In Section~\ref{sec:prel}, we briefly introduce some well-known results regarding the hull of linear codes. In Section~\ref{sec:1dh}, we characterize the properties of LCD codes and show that any LCD code over the finite extended binary field, with a minimum distance of at least 2, is equivalent to a one-dimensional hull of a linear code under a specific condition. Additionally, we present a construction of \((\ell + 1)\)-dimensional hulls of a linear code derived from \(\ell\)-dimensional hulls under a weak condition in Section~\ref{sec:1dhc}. We also provide a simple method for constructing a one-dimensional hull of a linear code from a given LCD code in the same section. Section~\ref{sec:11dhc} introduces various constructions of \(\ell\)-dimensional hulls of linear codes. The paper concludes in Section~\ref{sec: conclusion}.

\section{Some preliminaries}\label{sec:prel}
Let $\FF_q$ be a finite field with cardinality $q = p^m$, where $p$ is a prime number, and $m \geq 2$ is a positive integer. The set of all nonzero elements of $\FF_q$ is denoted as $\FF_q^* = \FF_q \setminus \{0\}$. 

Consider a linear code $\CC:= [n, k, d]$ in $\FF_q^n$ (with the standard notation) with length $n$, dimension $k$ and minimum distance $d$ (the metric employed here is the Hamming distance). For any two vectors $\mathbf{x} = (x_1, x_2, \ldots, x_n)$ and $\mathbf{y} = (y_1, y_2, \ldots, y_n)$ in $\FF_q^n$, the inner product is defined as 
$$\langle \mathbf{x}, \mathbf{y} \rangle = \sum_{i=1}^n x_i y_i.$$ 

The dual code of $\CC$ (denoted as $\CC^\perp$) is defined for a linear code $\CC$ over $\FF_q$ of length $n$ as 
$$\CC^\perp = \{\mathbf{x} \in \FF_q^n \,|\, \langle \mathbf{x}, \mathbf{c} \rangle = 0 \text{ for all } \mathbf{c} \in \CC\}.$$ 

A linear code $\CC$ is considered self-orthogonal if $\CC \subseteq \CC^\perp$, and it is self-dual if $\CC = \CC^\perp$. The hull of a linear code $\CC$ is defined as 
$$\text{hull}(\CC) = \CC \cap \CC^\perp.$$ 
If the dimension of the hull, $\dim(\text{hull}(\CC))$, is denoted as $\ell$, then $\CC$ is referred to as an $\ell$-dimensional hull code. In particular, if $\dim(\text{hull}(\CC)) = 0$, then $\CC$ is called an LCD code.

For any vector $\mathbf{a} = (a_1, a_2, \ldots, a_n) \in \FF_q^n$, we define the linear code $\CC_\mathbf{a}$ as 
$$\CC_\mathbf{a} = \{(a_1 c_1, a_2 c_2, \ldots, a_n c_n) \,|\, (c_1, c_2, \ldots, c_n) \in \CC\}.$$ 
The code $\CC_\mathbf{a}$ is linear if $\CC$ is linear.

A monomial matrix is an invertible matrix with rows of weight one (i.e., exactly one non-zero entry). If all the nonzero entries of a monomial matrix are ones, it is called a permutation matrix. Two codes $\mathcal{C}$ and $\mathcal{C}'$ over $\FF_q$ of the same length are said to be monomial equivalent (or simply equivalent) if there exists a monomial matrix $M$ such that $\mathcal{C}' = \mathcal{C}M = \{cM \,|\, c \in \mathcal{C}\}$.

It is well-known, according to MacWilliams' theorem (see [\cite{MacWilliams}, Theorem 4]), that every isometry on $\mathbb {F}_{q}$ with respect to the Hamming metric is represented by a monomial matrix. Since monomial equivalence preserves the weight distributions, equivalent codes maintain the same basic parameters: length, dimension, and minimum distance. It is straightforward to observe that the duals of equivalent codes are also equivalent. Specifically, if $\mathcal{C}_1$ and $\mathcal{C}_2$ are equivalent such that $\mathcal{C}_2 = \mathcal{C}_1M$, then $\mathcal{C}_{2}^{\perp}$ and $\mathcal{C}_{1}^{\perp}$ are equivalent as well.

Furthermore, if $\ba = (a_1, a_2, \ldots, a_n)$ where $a_i \in \FF_q^*$ for $1 \leq i \leq n$, then the code $\CC$ is equivalent to the code $\CC_\ba$, and the dual of $\CC_\ba$ is given by $(\CC_\ba)^\perp = \ba^{-1} \CC^\perp$, where $\ba^{-1} = (a_1^{-1}, a_2^{-1}, \ldots, a_n^{-1})$. For any permutation $\sigma$ of the set $\{1, 2, \ldots, n\}$, we define the linear code $\sigma(\CC)$ as follows:
\[
\sigma(\CC) = \{(c_{\sigma(1)}, c_{\sigma(2)}, \ldots, c_{\sigma(n)}) : (c_1, c_2, \ldots, c_n) \in \CC\}.
\]

If \( G \) is a generator matrix of \( \CC \), then the generator matrix of \( \CC_\ba \) is of the form 
\[
G_\ba  = \left(\begin{array}{ccccccc}
a_1 g_1 & a_2 g_2 & a_3 g_3 & \cdots & a_n g_n
\end{array}\right),
\]
where \( g_i \) are the columns of \( G \).

For a linear code \(\CC\) with a generator matrix \(G\) and a parity check matrix \(H\), the relationships between the ranks of the matrices \(GG^\top\) and \(HH^\top\) and the hull of \(\CC\) can be described as follows.

\begin{prop}\cite{Gue2018}\label{pr-2.1}
Let \(\CC := [n, k, d]\) be a linear code with generator matrix \(G\) and parity check matrix \(H\). The ranks of the matrices \(GG^\top\) and \(HH^\top\) are independent of the specific choices of \(G\) and \(H\), respectively. We have:
\begin{eqnarray*}
\text{rank}(GG^\top) & = & k - \dim(\hull(\CC)) = k - \dim(\hull(\CC^\perp)), \\
\text{rank}(HH^\top) & = & n - k - \dim(\hull(\CC)) = n - k - \dim(\hull(\CC^\perp)).
\end{eqnarray*}
\end{prop}

The following proposition completely characterises an LCD (Linear Complementary Dual) code.

\begin{prop}\cite{Mas92}\label{pr-2.2}
Let $\CC:=[n,k]$ be a linear code over $\FF_q$ with generator matrix $G$ and parity check matrix $H$. Then the following statements are equivalent
\begin{enumerate}
\item $\CC$ is LCD;
\item $k\times k$ matrix $GG^\top$ is invertible;
\item $(n-k)\times (n-k)$ matrix $HH^\top$ is invertible. 
\end{enumerate}
\end{prop}

The following proposition gives another characterization of an LCD code.
\begin{prop}\cite{Mas92}\label{pr-2.3}
Let $\CC:=[n,k]$ be a linear code over $\FF_q$ with generator matrix $G$ and parity check matrix $H$. Then $\CC$ is an LCD if and only if the matrix 
$\left(\begin{array}{cc} G \\ H \end{array}\right)$ is invertible.  
\end{prop}

\section{A General Construction of One-Dimensional Hull Codes over $\FF_{2^t}$}\label{sec:1dh}

This section presents a method for constructing linear codes with a one-dimensional hull. We aim to provide a result that helps control the dimension of the hull of a given linear code. 

Let $\CC := [n, k, d]$ be a linear code over $\FF_q$ with $q > 3$. According to \cite[Corollary 14]{CMT18}, the code $\CC$ is equivalent to an LCD (linear complementary dual) code. 

Assuming $\CC:= [n, k, d]$ is an LCD code with a generator matrix of the form: 
\[
G = \left(\begin{array}{cc} I_k & P \end{array}\right),
\]
we can state that if $\CC := [n, k, d]$ is a linear code with a minimum distance of $d \geq 2$, then there exists \( a \in \FF_q^* \) such that a generator matrix for $\CC$ (or an equivalent code to $\CC$) can be expressed as:
\begin{equation}\label{eq-3.1}
G = \left(\begin{array}{cccc}
1 & 0 & P_1 & a \\
0 & I_{k-1} & P_2 & b
\end{array}\right),
\end{equation}
where \( P_1 \) is a \( 1 \times (n - k - 1) \) matrix, \( P_2 \) is a \( (k - 1) \times (n - k - 1) \) matrix, and \( b \) is a \( (k - 1) \times 1 \) matrix.

From this point onward, we will consider the above form of the generator matrix. Let $\CC := [n, k, d]$ be a linear code over $\FF_q$ with \( q > 3 \), and assume it has a generator matrix \( G \) in the form given in \eqref{eq-3.1}. We will first demonstrate the existence of an equivalent code \( \tilde{\CC} \) with a certain generator matrix form.

\begin{lemma}\label{lm-3.1}
Let $\CC:=[n,k,d]$ be a linear code over $\FF_q$, where $q > 3$, with a minimum distance $d \geq 2$ and a generator matrix $G$ in the form given by equation~\eqref{eq-3.1}. If \( P_1 P_1^\top + a^2 = 0 \), then there exists an equivalent code \( \tilde{\CC} \) of \( \CC \) with a generator matrix of the form 
\[
\tilde{G} = \begin{pmatrix}
1 & 0 & P_1 & a_1 \\
0 & I_{k-1} & P_2 & b_1
\end{pmatrix}
\]
where \( P_1 P_1^\top + a_1^2 \neq 0 \).
\end{lemma}

\begin{proof}
Consider a non-zero element \( \mu \in \FF_q^* \) and define the code \( \tilde{\CC} = \{(c_1, c_2, \ldots, \mu c_n) \mid (c_1, c_2, \ldots, c_n) \in \CC\} \). This means \( \tilde{\CC} = \CC_\ba \) where \( \ba = (1, 1, \ldots, 1, \mu) \in (\FF_q^*)^n \). It can be observed that \( \tilde{\CC} \) is equivalent to \( \CC \), and a generator matrix \( \tilde{G} \) is of the form 
\[
\tilde{G} = \begin{pmatrix}
1 & 0 & P_1 & \mu a \\
0 & I_{k-1} & P_2 & \mu b
\end{pmatrix}.
\]
Next, we will prove that there exists an element \( \mu \in \FF_q^* \) such that \( P_1 P_1^\top + \mu^2 a^2 \neq 0 \).
 To ensure this condition holds, we choose  \( \mu \in \FF_q^* \) such that $\mu^2\neq 1$.
%By the assumption, \( P_1 P_1^\top + a^2 = 0 \pmod p \) implies \( P_1 P_1^\top +a^2= pl \), where \( p \) is the characteristic of \( \FF_q \) and \( l \in \mathbb{N} \). Since \( q > 3 \), there exists a non-zero \( \mu \in \FF_q^* \) such that \( \mu^2 - 1 \not\equiv 0 \pmod{p} \). 

%Therefore, for such \( \mu \), we have:
%\[
%P_1 P_1^\top + \mu^2 a^2 = pl - a^2 + \mu^2 a^2 = (\mu^2 - 1)a^2 + pl \equiv (\mu^2 - 1)a^2 \not\equiv 0 \pmod{p}.
%\]
This concludes the proof.
\end{proof}
\begin{lemma}\label{lm-3.2}
Let \( q > 3 \) be an even prime power. If \(\mathcal{C} := [n, k, d]\) is an LCD code over \(\FF_{q}\) with minimum distance \( d \geq 2 \) and a generator matrix \( G \) of the form given in \eqref{eq-3.1}, along with the condition \( P_1 P_1^\top + a^2 = 0 \), then \(\tilde{\mathcal{C}}\) (defined in Lemma~\ref{lm-3.1}) is also an LCD code.
\end{lemma}

\begin{proof}
Since \(\mathcal{C}\) is an LCD code, it follows from Proposition~\ref{pr-2.2} that \( HH^\top \) is invertible over \(\FF_{q}\). A parity check matrix \( H \) can be represented as follows:
\[
H = \begin{pmatrix}
-P_1^\top & -P_2^\top & I_{n-k-1} & 0 \\
-a & -b^\top & 0 & 1
\end{pmatrix}.
\]

Then, we can compute \( HH^\top \):
$$HH^\top=\left(\begin{array}{cc}
I_{n-k-1}+P_1^\top P_1 +P^\top_2P_2  & P_1^\top a+P^\top_2b  \\
aP_1+ b^\top P_2  &  a^2+b^\top b+1 
\end{array}\right).$$ 

By Lemma~\ref{lm-3.1}, there exists an element  $\mu \in \FF_q^*$ with $\mu^2\neq 1$ such that a generator matrix and a parity check matrix of $\tilde{\CC}$ are, respectively, of the form  $\tilde{G}$ 	and $\tilde{H}$ as below.
$$\tilde{G}=\left(\begin{array}{cccc}
1 & 0 & P_1 & \mu a \\
0 & I_{k-1} & P_2 & \mu b
\end{array}\right),~~~~~\tilde{H}=\left(\begin{array}{cccc}
-P_1^\top & -P_2^\top & I_{n-k-1} & 0 \\
-a & -b^\top & 0 & \mu^{-1}
\end{array}\right), \text{ respectively}.$$

Now consider 
$$\tilde{H}\tilde{H}^\top=\left(\begin{array}{cc}
I_{n-k-1}+P_1^\top P_1 +P^\top_2P_2  & P_1^\top a+P^\top_2b  \\
aP_1+ b^\top P_2  &  a^2+b^\top b+\mu^{-2} 
\end{array}\right).$$
 
 If $\tilde{H}\tilde{H}^\top$ is an invertible, the result follows immediately. Otherwise,
suppose \( \det(\tilde{H}\tilde{H}^\top) = 0 \).  Note that the equation for $\mu$ can have at most $2$ roots. But since q>3, thus, the proof is complete.%Since \( q > 3 \), we strategically select \( \mu \in \FF_q^*  \) with $\mu^2\neq 1$ to ensure that \( \tilde{H}\tilde{H}^\top \) remains invertible. By applying Proposition~\ref{pr-2.2}, we can conclude that \( \tilde{C} \) qualifies as an LCD code.
\end{proof}

\begin{thm}\label{th-3.1}
Let \( q > 3 \) be an even prime power. Let \(\mathcal{C}:=[n,k,d]\) be an LCD code over \(\FF_{q}\) with minimum distance \( d \geq 2 \), and let \( G \) be a generator matrix of the form 
\[
G = \left(\begin{array}{cccc}
1 & 0 & P_1 & a \\
0 & I_{k-1} & P_2 & b
\end{array}\right)\]
where \( P_1 P_1^\top + a^2 \neq 0 \). If \( P_1 P_2^\top + ab^\top = 0 \), then there exists \( \mathbf{a} \in (\FF_q^*)^n \) such that 
$\dim(\mathcal{C}_{\mathbf{a}} \cap \mathcal{C}_{\mathbf{a}}^\perp) = 1$.
\end{thm}

\begin{proof}
Consider a parity-check matrix \( H \) of \(\mathcal{C}\) in the form
\[
H = \begin{pmatrix}
-P_1^\top & -P_2^\top & I_{n-k-1} & 0 \\
-a & -b^\top & 0 & 1
\end{pmatrix}. 
\]
Since \(\mathcal{C}\) is an LCD code, by the well-known characterization of LCD codes given by Proposition~\ref{pr-2.3}, we have:

\begin{equation}\label{eq-3.3}
\textit{rank}\left(\begin{array}{cccc}
   G\\ H
\end{array}\right) = \textit{rank}\left(\begin{array}{cccc}
1 & 0 & P_1 & a \\
0 & I_{k-1} & P_2 & b\\
-P_1^\top & -P_2^\top & I_{n-k-1} & 0 \\
-a & -b^\top & 0 & 1
\end{array}\right) = n. 
\end{equation}

Therefore,
\begin{equation}\label{eq-3.6}
\textit{rank}\left(\begin{array}{cccc}
0 & I_{k-1} & P_2 & b\\
-P_1^\top & -P_2^\top & I_{n-k-1} & 0 \\
-a & -b^\top & 0 & 1
\end{array}\right) = n-1 \text{ and } \textit{rank}\left(\begin{array}{cccc}
1 & 0 & P_1 & a \\
0 & I_{k-1} & P_2 & b\\
-P_1^\top & -P_2^\top & I_{n-k-1} & 0 \\
\end{array}\right) = n-1.
\end{equation}

Establishing the existence of a one-dimensional hull of a linear code equivalent to $\CC$ involves demonstrating that there is an element $\ba \in (\FF_p^*)^n$ such that $\CC_{\ba}$ is equivalent to $\CC$ and $\dim(\CC_{\ba} \cap \CC_{\ba}^\perp) = 1.$

Denote 
$$\CC_{\lambda}=\{(\lambda c_1,c_2,\ldots,c_n)~|~(c_1,c_2,\ldots,c_n)\in \CC\}.$$ 
Then, let consider a generator matrix $G_\lambda$ and a parity check matrix $H_\lambda$ of $\CC_\lambda$ as
$$G_\lambda = \left(\begin{array}{cccc}
\lambda & 0 & P_1 & a \\
0 & I_{k-1} & P_2 & b
\end{array}\right) \text{ and } H_\lambda = \left(\begin{array}{cccc}
-\lambda^{-1}P_1^\top & -P_2^\top & I_{n-k-1} & 0 \\
-\lambda^{-1}a & -b^\top & 0 & 1
\end{array}\right), \text{ respectively}.$$
Note that since  $\CC_1 \subseteq \CC_\lambda$, $\CC_1 +\ \CC^\perp_\lambda \subseteq \CC_\lambda +\CC^\perp_\lambda$,  the linear code   $\CC_1$ is generated by 
$G_1 = \left(\begin{array}{cccc} 0 & I_{k-1} & P_2 & b \end{array}\right)$.
 
Moreover, $\CC_1 \cap \CC^\perp_\lambda = \{0\}$ as 
$\textit{rank}\left(\begin{array}{cccc}
0 & I_{k-1} & P_2 & b\\
-\lambda^{-1}P_1^\top & -P_2^\top & I_{n-k-1} & 0 \\
-\lambda^{-1}a & -b^\top & 0 & 1
\end{array}\right) = n-1$ (see Equation \eqref{eq-3.3}).

Since $\CC_1 + \CC^\perp_\lambda \subseteq \CC_\lambda + \CC^\perp_\lambda$, 
$\dim\left(\CC_1 + \CC^\perp_\lambda\right) \leq \dim\left(\CC_\lambda + \CC^\perp_\lambda\right) \implies
\dim(\CC_1) + \dim(\CC^\perp_\lambda) - \dim(\CC_1 \cap \CC^\perp_\lambda)
\leq \dim(\CC_\lambda) + \dim(\CC^\perp_\lambda) - \dim(\CC_\lambda \cap \CC^\perp_\lambda) \implies \dim(\CC_\lambda \cap \CC^\perp_\lambda) \leq 1$.

We will demonstrate that there exists a nonzero codeword in the intersection of $\CC_\lambda$ and $\CC^\perp_\lambda$ for an appropriate choice of $\lambda \in \FF_q^*$. To establish this, it suffices to prove that there exists a suitable  $\lambda \in \FF_q^*$, $\det(G_\lambda G_\lambda^\top) = 0$. 

Now, $$G_\lambda G_\lambda^\top = \left(\begin{array}{cccc}
\lambda^2+P_1P_1^\top+a^2 & P_1P_2^\top+ab^\top\\
P_2P_1^\top+ba & I_{k-1}+P_2P_2^\top+bb^\top
\end{array}\right) = \left(\begin{array}{cccc}
\lambda^2+P_1P_1^\top+a^2 & 0 \\
0 & I_{k-1}+P_2P_2^\top+bb^\top
\end{array}\right).$$
Since $q$ is even and $P_1P_1^\top + a^2 \neq 0$, there always exists $\lambda \in \FF_q^*$ such that $\lambda^2 + P_1P_1^\top + a^2 = 0$.
That implies, for such $\lambda$, $\det(G_\lambda G_\lambda^\top)=0$, which completes  the proof.
\end{proof}

\begin{example}
Let \(\FF_8\) be the finite field with a primitive root \(\omega\). Consider the linear code \(\CC := [10,3,7]\) with the following generator matrix:
\[
G = \left(\begin{array}{ccccccccccc}
1 & 0 & 0 & \omega^4 & \omega^4 & \omega^5 & \omega^2 & \omega & \omega^5 & \omega^4 \\
0 & 1 & 0 & 0 & \omega^5 & \omega^5 & \omega^4 & \omega^4 & 1 & 1 \\
0 & 0 & 1 & \omega^3 & \omega^6 & \omega^2 & 1 & 0 & \omega^4 & \omega^5 \\
\end{array}\right).
\]
Here, we have \(GG^\top = \left(\begin{array}{cccc}
1+\omega^3 & 0 & 0 \\
0 & 1+\omega^2 & 0 \\
0 & 0 & 1
\end{array}\right)\), and the generator matrix \(G\) satisfies all conditions of Theorem~\ref{th-3.1}. 

According to Theorem~\ref{th-3.1}, we choose 
\[
\ba = \left(\begin{array}{cccccccccc} \omega^5 & 1 & 1 & 1 & 1 & 1 & 1 & 1 & 1 & 1 \end{array} \right)
\]
such that the code \(\CC_\ba\) has a one-dimensional hull and is also equivalent to \(\CC\).
\end{example}

\section{Construction of an $(\ell + 1)$-dimensional Hull from an $\ell$-dimensional Hull}\label{sec:1dhc}

This section presents a method for deriving an $(\ell + 1)$-dimensional hull from a given $\ell$-dimensional hull, where $\ell$ is a nonnegative integer. We begin by recalling some mathematical concepts from linear algebra. 

\begin{prop}\label{p-11}
Let $\CC:= [n,k]$ be a linear code defined over the finite field $\FF_q$, with generator matrix \( G \) and parity check matrix \( H \). The dimension of the intersection of \( \CC \) and its dual \( \CC^\perp \) satisfies \(\dim(\CC \cap \CC^\perp) = \ell\) if and only if the rank of the matrix formed by stacking \( G \) and \( H \) is \( n - \ell \):
\[
rank\left(\begin{array}{cc} G \\ H \end{array}\right) = n - \ell.
\]
\end{prop}

\begin{proof}
We know that a vector \( x \in \mathcal{C} \) if and only if \( Hx^T = 0 \), and \( x \in \mathcal{C}^\perp \) if and only if \( Gx^T = 0 \), where \( G \) and \( H \) are the generator and parity-check matrices of \( \mathcal{C} \), respectively. Therefore, \( x \in \mathcal{C} \cap \mathcal{C}^\perp \) if and only if \( Ax^T = 0 \), where 
\[
A = \begin{pmatrix} G \\ H \end{pmatrix}.
\]

This implies that \( \dim(\mathcal{C} \cap \mathcal{C}^\perp) = \ell \) if and only if 
\[
\text{rank}\begin{pmatrix} G \\ H \end{pmatrix} = n - \ell.
\]
\end{proof}

Let \(\mathcal{C}:=[n,k]\) be a linear code over \(\FF_{q}\) with a generator matrix \(G\) and a parity check matrix \(H\). Then, \(\dim(\mathcal{C} \cap \mathcal{C}^\perp) = 1\) if and only if \(rank\left(\begin{array}{cc} G \\ H  \end{array}\right) = n-1\).

\begin{example}\label{ex:lc1dh}
Consider the linear code \(\mathcal{C}:=[n,k]\) over \(\FF_{2^m}\) with the generator matrix
\[
G = \left(\begin{array}{cc} I_k & P \end{array}\right),
\]
where \(PP^\top=\left(a_{i,j}\right)_{1\leq i,j \leq k}\) and \(a_{i,j}=1\). 

Now, we have
\[
GG^\top = I_k + PP^\top = (a_{i,j}),
\]
where \(a_{i,i}=0\) and \(a_{i,j}=1\) for \(i \neq j\). If \(k\) is odd, then \(rank(GG^\top) = k-1\). 

Consequently,
\[
rank\left(\begin{array}{cc} G \\ H  \end{array}\right) = rank\left(\begin{array}{cc} I_k + PP^\top & 0  \\ P^\top & I_{n-k} \end{array}\right) = n-1.
\]

Specifically, let \(\mathcal{C}:=[6,3,3]\) be a linear code 
over $\FF_{4}$ with generator matrix
$G=\left(\begin{array}{cc} I_k & P\\
\end{array}\right)= \left(\begin{array}{cccccc} 1 & 0& 0&1&1&1\\
0 & 1& 0&\omega&\omega^2&0\\
0 & 0& 1&0&\omega&\omega^2\\
\end{array}\right)$, where $\omega$ is a primitive root of $\FF_4$. It is easy to see that $PP^\top=\left(\begin{array}{ccc} 1 & 1 & 1\\
1 & 1 & 1\\
1 & 1 & 1\\
\end{array}\right)$. The matrix $GG^\top = I_{3}+PP^\top$ is of rank $2$, and $\CC$ is a linear code with a one-dimensional hull.
\end{example}

For \( v = (v_1, v_2, \ldots, v_n) \in \FF_q^{n} \), the support set of \( v \) is defined as 
\[
I(v) = \{i \mid 1 \leq i \leq n \text{ and } v_i \neq 0\}
\]
and the Hamming weight of \( v \) is defined as 
\[
\wt(v) = |I(v)|.
\]

Let \( M \) be an \( n \times n \) matrix over \( \FF_q \) and let \( \mathbf{u} \) be a non-zero element of \( \FF_q^{n} \) such that \( \wt(\mathbf{u}) = w \) with support set \( I = I(\mathbf{u}) = \{i_1, i_2, \ldots, i_w\} \). We define \( \text{diag}_n(\mathbf{u}) \) as an \( n \times n \) diagonal matrix with diagonal elements \( u_1, u_2, \ldots, u_n \).

According to the paper \cite{CMT18}, we define \( M_I \) as an \((n-w) \times (n-w)\) submatrix of \( M \) by deleting the \( i_1 \)-th, \( i_2 \)-th, \ldots, \( i_w \)-th columns and rows of \( M \). If \( I = \{1, 2, \ldots, n\} \), then we set \( M_I = \left( 1 \right) \), and if \( I = \emptyset \), we have \( M_{\emptyset} = M \).

In the sequel, we present Lemma~\ref{lm-3ab}, which generalizes \cite[Theorem 11]{CMT18}. Given a hull of dimension \( \ell \) of a linear code over \( \FF_q^2 \) (where \( q \) is a prime), it is possible to construct an equivalent hull of dimension \( \ell' \) for a linear code, where \( 0 \leq \ell' \leq \ell \). We extend the result from \cite[Lemma 9]{CMT18} (as presented in Lemma~\ref{lm-3a}) to construct an equivalent hull of dimension \( \ell' \) for a linear code where \( \ell-1 \leq \ell' \leq \ell+1 \) over the field \( \FF_q \), where \( q \) is a power of \( 2 \).

\begin{lemma}\cite[Lemma 9]{CMT18}\label{lm-3a}
For an integer \( t \) satisfying \( 0 \leq t \leq n - 1 \), let \( M \) be an \( n \times n \) matrix over \( \FF_{q} \) such that \( \det(M_I) = 0 \) for any subset \( I = \{1, 2, \ldots, n\} \) with \( 0 \leq |I| \leq t \). Assume that \( \mathbf{u} \) is a non-zero element of \( \FF_q^n \) satisfying \( 1 \leq \text{wt}(\mathbf{u}) \leq t + 1 \) with the supporting set \( J \). Then, we have:
$$\det(M+\text{diag}_n (\textbf{u}))=\left(\prod_{j\in J} u_j \right)\det(M_J).$$
\end{lemma}

\begin{lemma}\label{lm-3ab} 
Let \( \mathcal{C} \) be a linear code over \( \FF_q \) of length \( n \) with dimension \( k \), where \( q > 3 \) is a prime power. Assume that \( \dim(\mathcal{C} \cap \mathcal{C}^\perp) = \ell \). 

Then, there exists an index \( j \) such that \( 1 \leq j \leq n \), and \( \mathbf{a} = (a_1, a_2, \ldots, a_n) \in \FF_q^n \), with \( a_i = 1 \) for all \( i \neq j \) and \( a_j \in \FF_q \setminus \{0, 1\} \), satisfying 
\[
|\dim(\mathcal{C}_\mathbf{a} \cap \mathcal{C}_\mathbf{a}^\perp) - \dim(\mathcal{C} \cap \mathcal{C}^\perp)| \leq 1.
\]
\end{lemma}

\begin{proof}
Without loss of generality, let \( G \) be a generator matrix of \( \mathcal{C} \) of the form 
$G= \left(\begin{array}{cccc}
 I_k& P
\end{array}\right)$. For $\ba=(a_1,a_2,\ldots,a_n) \in (\FF_q^*)^n$, a generator matrix of the code 
$\CC_{\ba}=\{(a_1c_1,a_2c_2,\ldots,a_nc_n)~|~(c_1,c_2,\ldots,c_n)\in \CC\}$
is $G_\ba  = \left(\begin{array}{ccccccc}
a_1g_1 & a_2g_2 & a_3g_3 &\cdots & a_ng_n
\end{array}\right),$
where $g_i$'s are the columns of $G$.

We have  \(\det(G_{\ba}G_{\ba}^\top) = \det(GG^\top + \text{diag}_k(\textbf{u}))\), where \(u_j = a_j^2 - 1\) for \(1 \leq j \leq n\). By applying Lemma~\ref{lm-3a}, we find that \(\det(G_{\ba}G_{\ba}^\top) = u_{j} \det((GG^\top)_j)\) for some \(j\) where \(1 \leq j \leq n\). Here, \((GG^\top)_j\) denotes the submatrix of \(GG^\top\) obtained by deleting the \(j\)-th row and column. According to our hypothesis, \(\text{diag}_{k}(\textbf{u})\) acts on only one element of \(GG^\top\). Consequently, we can deduce that 
\[
\dim(\CC \cap \CC^\perp) - 1 \leq \dim(\CC_{\ba} \cap \CC_{\ba}) \leq \dim(\CC \cap \CC^\perp) + 1.
\]
This implies that 
\[
|\dim(\CC_\ba \cap \CC_\ba^\perp) - \dim(\CC \cap \CC^\perp)| \leq 1.
\]
\end{proof}
Hence, for a given code \(\CC\) with \(\dim(\CC \cap \CC^\perp) = \ell\), using a similar technique to the proof of \cite[Theorem 11]{CMT18} and Lemma~\ref{lm-3ab}, we can deduce a code \(\CC_\ba\) such that \(\dim(\CC_\ba \cap \CC_\ba^\perp) = \ell - 1\). In other words, it is possible to construct a \((\ell-1)\)-dimensional hull from a given \(\ell\)-dimensional hull.

In general, designing codes with hulls with decreasing dimensions is not difficult. In the following sections, however, our goal is quite the opposite: we aim to construct an \((\ell + 1)\)-dimensional hull of a code over \(\FF_q\) (where \(q\) is a power of \(2\)) from a given \(\ell\)-dimensional hull of a code, for \(0 \leq \ell < \min\{k, n-k\}\). A construction over the field \(\FF_2\) is presented in \cite{LSK23}.

%\textcolor{red}{This is also done in particular via the building-up construction. See, for example, the paper  "Characterization and Construction of optimal binary linear codes with one-dimensional hull." by 
%Shitao Li,  Minjia Shi and  Jon-Lark Kim (June 2023)}
Let $\CC := [n,k,d]$ be a linear code with $\dim(\CC\cap\CC^\perp)=\ell$, where $0\leq \ell < \min\{k,n-k\}$. Without loss of generality, a generator matrix of $\CC\cap\CC^\perp$ is of the form $G_{\ell}=\left(\begin{array}{ccccc}
I_{\ell} & Q_1 & Q_2 & P_1   
\end{array}\right)$.
Let $r_i, 1 \leq i \leq \ell$, be the rows of $G_{\ell}$.  We can extend the linearly independent set $\{r_1,r_2,\ldots,r_\ell\}$ to a basis $\{r_1,r_2,\ldots,r_\ell,r'_{\ell+1},\ldots,r'_{k}\}$ of $\CC$. Without loss of generality, we consider a generator matrix of  $\CC$ to be of the form
$G=\left(\begin{array}{ccccc}
I_{\ell} & Q_1 & Q_2 & P_1  \\
0 & \alpha & 0 & P_2 \\
0 & 0 &I_{k-\ell-1} & P_3 
\end{array}\right)$, where $\alpha\in\FF_q^*$.\\
 Denote $G_{k-\ell}=\left(\begin{array}{ccccc}
 \alpha & 0 & P_2 \\
 0 &I_{k-\ell-1} & P_3 
\end{array}\right)$. Then 
\begin{equation}\label{eq-3.s31q}
\det(G_{k-\ell}G_{k-\ell}^\top)=\alpha^{2}\det(I_{k-\ell-1}+P_3P_3^\top) + \beta'.
\end{equation}
Now setting the notation $\beta' = \det(G_{k-\ell}G_{k-\ell}^\top)-\alpha^{2}\det(I_{k-\ell-1}+P_3P_3^\top)$, we show the existence of an $\ell+1$-dimensional hull of a linear code equivalent to $\CC$ where its generator matrix $G$ satisfies $rank(GG^\top) = k-\ell$ (that is, $\CC$  is an $\ell$-dimensional hull) in the following theorem.

\begin{thm}\label{th-3.22q}
Let $q > 3$ be an even prime power. Let $\CC:=[n,k,d]$ be a linear code over $\FF_q$ with a generator matrix 
$G=\left(\begin{array}{ccccc}
I_{\ell} & Q_1 & Q_2 & P_1  \\
0 & \alpha & 0 & P_2 \\
0 & 0 &I_{k-\ell-1} & P_3 
\end{array}\right)$, where $\alpha\in\FF_q^*$ and  $G_{\ell}=\left(\begin{array}{ccccc}
I_{\ell} & Q_1 & Q_2 & P_1  
\end{array}\right)$ is a generator matrix of $\CC\cap \CC^\perp$. Assume $\det(I_{k-\ell-1}+P_3P_3^\top)\neq 0$, $Q_1=0$, and $\beta' \neq 0$ (defined in Equation~\eqref{eq-3.s31q}). Then there exists $\ba \in (\FF_q^*)^n$ such that $\dim(\CC_\ba \cap \CC_\ba^\perp) =\ell+ 1$.  
\end{thm}
\begin{proof}
	
		Given the generator matrix of the code $\CC$ as
		\[
		G = \begin{pmatrix}
			I_{\ell} & Q_1 & Q_2 & P_1 \\
			0 & \alpha & 0 & P_2 \\
			0 & 0 & I_{k-\ell-1} & P_3
		\end{pmatrix},
		\]
		we observe that $ G $  is row equivalent to
		\[
		G = \begin{pmatrix}
			I_{\ell} & 0 & 0 & R_1 \\
			0 & 1 & 0 & \alpha^{-1} P_2 \\
			0 & 0 & I_{k-\ell-1} & P_3
		\end{pmatrix},
	   \]
		where $ R_1 $ is obtained after suitable row operations.
		
		Now, consider a parity-check matrix $ H $ of the code $\CC$, defined by
		\[
		H = \begin{pmatrix}
			-R_1^\top & -\alpha^{-1} P_2^\top & -P_3^\top & I_{n-k}
		\end{pmatrix}.
		\]
		This matrix is row equivalent to the following block matrix form
		\[
		H = \begin{pmatrix}
			I_{\ell} & Q_1 & Q_2 & P_1 \\
			P_1'^\top & -\alpha^{-1} P_2'^\top & -P_3'^\top & I_{n-k-\ell}
		\end{pmatrix}.
		\]

Since $\dim(\CC\cap\CC^\perp)=\ell$, by Proposition~\ref{p-11}, 
\begin{equation}\label{eq-4.5}
\textit{rank}\left(\begin{array}{cccc} G \\  H \end{array}\right) 
= \textit{rank}\left(\begin{array}{cccc} 
I_{\ell} & Q_1 & Q_2 & P_1  \\
0 & \alpha & 0 & P_2 \\
0 & 0 &I_{k-\ell-1} & P_3\\ 
I_{\ell} & Q_1 & Q_2 & P_1  \\
{P'}_1^\top & -\alpha^{-1} {P'}_2^\top &-{P'}_3^\top & I_{n-k-\ell}  
\end{array}\right) = n-\ell. 
\end{equation}
Then
\begin{eqnarray}\label{eq-4.4}
\textit{rank}\left(\begin{array}{cccc}
I_{\ell} & Q_1 & Q_2 & P_1  \\
0 & 0 &I_{k-\ell-1} & P_3\\ 
{P'}_1^\top & -\alpha^{-1} {P'}_2^\top &-{P'}_3^\top & I_{n-k-\ell}  
\end{array}\right) & = & n-\ell-1 \text{ and } \nonumber\\
\textit{rank}\left(\begin{array}{cccc}
0 & \alpha & 0 & P_2 \\
0 & 0 &I_{k-\ell-1} & P_3\\ 
I_{\ell} & Q_1 & Q_2 & P_1  \\
{P'}_1^\top & -\alpha^{-1} {P'}_2^\top &-{P'}_3^\top & I_{n-k-\ell} 
\end{array}\right) & = & n-\ell.
\end{eqnarray}

To demonstrate the existence of an $\ell+1$-dimensional hull of a linear code equivalent to $\CC$, we aim to show that for any $\ba \in (\FF_q^*)^n$, the code $\CC_{\ba}$ is equivalent to $\CC$ and that $\dim\left(\CC_{\ba}\cap\CC_{\ba}^\perp\right) = \ell + 1$.
For $\lambda \in \FF_q^*$, denote
$$\CC_{\lambda} = \{( c_1,c_2,\ldots,c_\ell,\lambda c_{\ell+1},c_{\ell+2},\ldots,c_n)~|~(c_1,c_2,\ldots,c_n)\in \CC\}.$$
 Since $Q_1=0$, then we consider a generator matrix $G_\lambda$ and a parity check matrix $H_\lambda$ of $\CC_\lambda$ as
$$G_\lambda = \left(\begin{array}{cccc}
I_{\ell} & 0 & Q_2 & P_1  \\
0 & \lambda\alpha & 0 & P_2 \\
0 & 0 &I_{k-\ell-1} & P_3
\end{array}\right) \text{ and } 
H_\lambda = \left(\begin{array}{cccc}
I_{\ell} & 0 & Q_2 & P_1  \\
{P'}_1^\top & -\lambda^{-1}\alpha^{-1} {P'}_2^\top &-{P'}_3^\top & I_{n-k-\ell}
\end{array}\right), \text{ respectively}.$$ 
Let us consider $\CC_1$ as a linear code generated by 
$G_1 = \left(\begin{array}{cccc} 0 & 0 &I_{k-\ell-1} & P_3 \end{array}\right)$.
Moreover, $\CC_1 + \CC^\perp_\lambda \subseteq \CC_\lambda + \CC^\perp_\lambda$ and 
$\CC_1 \cap \CC^\perp_\lambda = \{0\}$ (since
$\textit{rank}\left(\begin{array}{cccc}
0 & 0 &I_{k-\ell-1} & P_3 \\
I_{\ell} & Q_1 & Q_2 & P_1  \\
{P'}_1^\top & -\lambda^{-1}\alpha^{-1} {P'}_2^\top &-{P'}_3^\top & I_{n-k-\ell}
\end{array}\right) = n-\ell-1$ from  Equation~\eqref{eq-4.4}).
From the fact  $\CC_1 + \CC^\perp_\lambda \subseteq \CC_\lambda + \CC^\perp_\lambda$, one has:
$$\dim\left(\CC_1 + \CC^\perp_\lambda\right) \leq \dim\left(\CC_\lambda + \CC^\perp_\lambda\right).$$ Therefore,  
$$\dim(\CC_1) + \dim(\CC^\perp_\lambda) - \dim(\CC_1 \cap \CC^\perp_\lambda)
\leq \dim(\CC_\lambda) + \dim(\CC^\perp_\lambda) - \dim(\CC_\lambda \cap \CC^\perp_\lambda).$$ Consequently,  $$\dim(\CC_\lambda \cap \CC^\perp_\lambda) \leq \ell+ 1.$$
 Furthermore, $ \left(\begin{array}{cccc}
I_{\ell} & 0 & Q_2 & P_1  \\
\end{array}\right)$ is a part of
$H_\lambda = \left(\begin{array}{cccc}
I_{\ell} & 0 & Q_2 & P_1  \\
{P'}_1^\top & -\lambda^{-1}\alpha^{-1} {P'}_2^\top &-{P'}_3^\top & I_{n-k-\ell}
\end{array}\right)$, one gets  $\textit{rank} \left(\begin{array}{cc}
G_{\lambda}\\
H_{\lambda} 
\end{array}\right)\leq n-\ell.$ By Proposition \ref{p-11}, we have $\dim(\CC_\lambda \cap \CC^\perp_\lambda)\geq \ell$. Hence, $\ell\leq \dim(\CC_\lambda \cap \CC^\perp_\lambda)\leq \ell+1$.

For the simplicity, set $G_\lambda=\left(\begin{array}{cc}
I_{\ell} & P\\
0 & G'_{\lambda} 
\end{array}\right),$ where $P= \left(\begin{array}{cccc}
  0 & Q_2 & P_1  \\
\end{array}\right)$ and $G'_{\lambda}= \left(\begin{array}{cccc}
 \lambda\alpha & 0 & P_2 \\
 0 &I_{k-\ell-1} & P_3
\end{array}\right)$. 
Further, 
$$G_\lambda G_\lambda^\top = \left(\begin{array}{cccc}
0 & 0 & 0 \\
0& \lambda^{2}\alpha^{2}+P_2P_2^\top &P_2P_3^\top \\
0 & P_3P_2^\top & I_{k-\ell-1}+P_3P_3^\top
\end{array}\right)=\left(\begin{array}{cccc}
0 & 0 \\
0& G'_{\lambda}{G'}_{\lambda}^\top
\end{array}\right),$$ as $\left(\begin{array}{ccccc}
I_{\ell} & 0 & Q_2 & P_1  \\
\end{array}\right)$ is a part of generator matrix of $\CC_{\lambda}\cap \CC_{\lambda}^\perp$. Alternatively, it is enough to show that for a suitable non-zero $\lambda$,
$\det(G'_\lambda {G'}_\lambda^\top)=0$.
Here, $\det(G_{k-\ell}G_{k-\ell}^\top)=\alpha^{2}\det(I_{k-\ell-1}+P_3P_3^\top) + \beta'$. Then from the form of $G'_{\lambda}{G'}_{\lambda}^\top$, we have $\det(G'_{\lambda}{G'}_{\lambda}^\top) = \lambda^2\alpha^2\det(I_{k-\ell-1}+P_3P_3^\top) + \beta'$. Since $q$ is even and by the hypothesis $\det(I_{k-\ell-1}+P_3P_3^\top) \neq 0 $ and $\beta' \neq 0$, there always exists 
$\lambda \in \FF_q^*$ such that $\lambda^2\alpha^{2} \det(I_{k-1}+P_2P_2^\top) + \beta' = 0 \implies \det(G'_{\lambda}{G'}_{\lambda}^\top) = 0$. 
\end{proof}

Let $\CC$ be a linear code with a generator matrix 
$G=\left(\begin{array}{cccc}
\alpha & 0 & P_1  \\
0 & I_{k-1} & P_2 
\end{array}\right)$, where $\alpha\in\FF_q^*$. Then 
\begin{equation}\label{eq-3.31}
\det(GG^\top) = \alpha^2 \det(I_{k-1}+P_2P_2^\top) + \beta.
\end{equation}
Further, setting the notation $\beta = \det(GG^\top) - \alpha^2 \det(I_{k-1}+P_2P_2^\top)$, we show the existence of a one-dimensional hull of a linear code equivalent to $\CC$ where its generator matrix $G$ satisfies that $GG^\top$ is invertible (that is, $\CC$  is an LCD code which implies $\ell = 0$) in the following corollary.
\begin{cor}\label{th-3.22}
Let $q > 3$ be an even prime power. Let $\CC:=[n,k,d]$ be a linear code over $\FF_q$ with a generator matrix 
$G = \left(\begin{array}{cccc}
\alpha & 0 & P_1  \\
0 & I_{k-1} & P_2 
\end{array}\right)$, where $\alpha\in\FF_q^*$.
Assume $\det(GG^\top) \neq 0$  and 
$\beta \neq 0$. Then there exists 
$\ba \in (\FF_q^*)^n$ such that $\dim(\CC_\ba \cap \CC_\ba^\perp) = 1$.  
\end{cor}
\begin{example}
Let $\CC:=[4,2,2]$ be a linear code over $\FF_4$ with generator matrix 
$G = \left(\begin{array}{cccc}
1 & 0 & 1 & 0 \\
0 & 1 & w^2 & w^2 
\end{array}\right)$, where $w$ is a primitive root of $\FF_4$. 
Here, $\det(GG^\top) = 1 + w^2$. Then, by Corollary~\ref{th-3.22}, the code $\CC$ is equivalent to a linear code $\CC_{\lambda}$ with one-dimensional hull. Consider a generator matrix $G_{\lambda}$ of $\CC_\lambda$ of the form
$G_\lambda = \left(\begin{array}{cccc}
w & 0 & 1 & 0 \\
0 & w^2 & w^2 & 1 
\end{array}\right).$ Then by Theorem \ref{th-3.22q}, we have there exists 
$\ba = \left(\begin{array}{cccc} 1 & w^2 & 1& 1\end{array}\right)$ such that 
$\CC_\ba$ has a two-dimensional hull. Note that $\CC$ is equivalent to a self-dual code.
\end{example}

Let \( x \) be a transcendental element over \( \FF_q \), where \( q \geq 3 \). Define \( \mathcal{P}_k = \{ f \in \FF_q[x] : \deg(f) \leq k - 1 \} \). We consider the corresponding function field \( \FF_q(x) \) associated with \( \FF_q[x] \). For a vector \( \textbf{b} = (\alpha_1, \alpha_2, \ldots, \alpha_n) \), where the \( \alpha_i \) are distinct elements in \( \FF_q \) (with \( n \leq q \)), we define the code \( RS_k(\textbf{b}) = \{ (f(\alpha_1), f(\alpha_2), \ldots, f(\alpha_n)) \mid f \in \mathcal{P}_k \} \) as a \( k \)-dimensional Reed-Solomon code. 

Let \( h(x) = \prod_{i=1}^n (x - \alpha_i) \), and denote the derivative of \( h \) with respect to \( x \) as \( h'(x) \). We associate \( RS_k(\textbf{b}) \) with the parameters 
\[
RS_k(\textbf{b}) = [n, k, n-k+1]
\]
Its dual code is given by 
\[
RS_k(\textbf{b})^\perp = [n, n-k, k+1].
\] 

At the end of this section, we shall construct a one-dimensional hull from a given linear code with a linear complementary dual (LCD) property. Specifically, we will create an extended code from an existing LCD code and determine the one-dimensional hull of this extended code.

\begin{construction}\label{con-1}
Let \( \FF_{2^t} \) be an extended binary field. We define \( RS_k(\textbf{b}) = \{ (f(\alpha_1), f(\alpha_2), \ldots, f(\alpha_n)) \mid f \in \mathcal{P}_k \} \) as a \( k \)-dimensional Reed-Solomon code, where \( n \leq 2^t - 1 \) and the \( \alpha_i \) are distinct elements in \( \FF_q^* \). Thus, we have \( RS_k(\textbf{b}) = [n, k, n-k+1] \), which is an MDS code. Since \( 2^t \geq 4 \), by \cite[Corollary 14]{CMT18}, the code \( [n, k, n-k+1] \) is an LCD code. 

Let \( G \) be a generator matrix of the code \( [n, k, n-k+1] \) such that \( \det\left(GG^\top\right) \neq 0 \). We will construct a new code \( \tilde{\CC} \) with a generator matrix defined as 
\[
\tilde{G} = \begin{pmatrix}
\alpha & \mathbf{P} \\
0 & G 
\end{pmatrix}
\]
where \( \det\left(\tilde{G}\tilde{G}^\top\right) \neq \det\left(GG^\top\right) \), \( \mathbf{P} = (\beta_1, \beta_2, \ldots, \beta_n) \) with \( \beta_i \in \FF_q \), and \( \alpha \in \FF_q^* \). 

If \( \tilde{G}\tilde{G}^\top = 0 \), then the code \( \tilde{\CC} \) is given by \( [n+1, k] \) and has a one-dimensional hull. Also, if \( \tilde{G}\tilde{G}^\top \neq 0 \), then by Corollary~\ref{th-3.22}, we conclude that \( \tilde{\CC} \) is equivalent to a linear code \( \tilde{\CC}_{\lambda} \) that possesses a one-dimensional hull.
\end{construction}

We will now provide some numerical examples of Construction~\ref{con-1}.

\begin{example}
Let $\omega$ is a primitive root of $\FF_4$ and $\CC:=[3,1,3]$ be a linear code over $\FF_4$ with generator matrix 
$G = \left(\begin{array}{cccc}
1 & 1 & 1  
\end{array}\right)$. Now, by Construction \ref{con-1}, we define $\tilde{\CC}$ with generator matrix
$\tilde{G}  = \left(\begin{array}{cccc}
\omega^2 &\omega&1&0\\
0& 1 & 1 & 1  
\end{array}\right).$ It is easy to see that $\det\left(\tilde{G}\tilde{G}^\top\right)\neq \det\left(GG^\top\right)$. However, $\det\left(\tilde{G}\tilde{G}^\top\right)=\omega.$ Hence, $\tilde{\CC}=[4,2,3]$ is equivalent to a linear code  with one-dimensional hull.
\end{example}

\begin{example}
Let $\omega$ is a primitive root of $\FF_8$ and $\CC_1:=[7,1,7]$, $\CC_2:=[7,3,5]$, and $\CC_3:=[7,5,2]$ be three linear codes over $\FF_4$ with generator matrices 
$G_1 = \left(\begin{array}{ccccccc}
1 & 1 & 1 &1 & 1 & 1 &1  
\end{array}\right)$, \\
$G_2 = \left(\begin{array}{ccccccc}
1 & 1 & 1 &1 & 1 & 1 &1 \\
1 & \omega & \omega^2 & \omega^3 & \omega^4 & \omega^5 & \omega^6\\
1& \omega^6 & \omega^5 & \omega^4 & \omega^3 & \omega^2 & \omega 
\end{array}\right)$ and $G_3 = \left(\begin{array}{ccccccc}
1 & 1 & 1 &1 & 1 & 1 &1 \\
1 & \omega & \omega^2 & \omega^3 & \omega^4 & \omega^5 & \omega^6\\
1& \omega^6 & \omega^5 & \omega^4 & \omega^3 & \omega^2 & \omega \\
1& \omega^2 & \omega^4 & \omega^6 & \omega & \omega^3 & \omega^5 \\
1& \omega^5 & \omega^3 & \omega & \omega^6 & \omega^4 & \omega^2
\end{array}\right)$, respectively. Now, by Construction \ref{con-1}, we define $\tilde{\CC}_1$, $\tilde{\CC}_2$ and $\tilde{\CC}_3$ with generator matrices\\
$\tilde{G}_1  = \left(\begin{array}{cccccccc}
\omega^5&\omega^6&\omega^3&\omega^4&\omega &\omega^2&1&0\\
0& 1 & 1 & 1 &1 & 1 & 1 &1  
\end{array}\right)$, $\tilde{G}_2  = \left(\begin{array}{cccccccc}
\omega^5&\omega^6&\omega^3&\omega^4&\omega &\omega^2&1&0\\
0& 1 & 1 & 1 &1 & 1 & 1 &1 \\
0& 1 & \omega & \omega^2 & \omega^3 & \omega^4 & \omega^5 & \omega^6\\
0& 1& \omega^6 & \omega^5 & \omega^4 & \omega^3 & \omega^2 & \omega 
\end{array}\right)$ and \\
$\tilde{G}_2  = \left(\begin{array}{cccccccc}
\omega^5&\omega^6&\omega^3&\omega^4&\omega &\omega^2&1&0\\
0& 1 & 1 & 1 &1 & 1 & 1 &1 \\
0& 1 & \omega & \omega^2 & \omega^3 & \omega^4 & \omega^5 & \omega^6\\
0& 1& \omega^6 & \omega^5 & \omega^4 & \omega^3 & \omega^2 & \omega \\
0& 1& \omega^2 & \omega^4 & \omega^6 & \omega & \omega^3 & \omega^5 \\
0& 1& \omega^5 & \omega^3 & \omega & \omega^6 & \omega^4 & \omega^2
\end{array}\right)$, respectively. It is easy to see that $\det\left(\tilde{G}_i\tilde{G}_i^\top\right)\neq \det\left(G_iG_i^\top\right)$ and $\det\left(\tilde{G}_i\tilde{G}_i^\top\right)=\omega^3$ for $1\leq i \leq 3$. Hence, $\tilde{\CC}_1=[8,2,7]$, $\tilde{\CC}_2=[8,4,4]$ and $\tilde{\CC}_3=[8,6,2]$ are equivalent to linear codes  with one-dimensional hull.
\end{example}
\begin{example}
Let $\omega$ is a primitive root of $\FF_{16}$ and $\CC_1:=[15,1,15], \CC_2:=[15,3,13], \CC_3:=[15,5,11]$, $\CC_4:=[15,7,9], \ldots,\CC_7:=[15,15,1]$  be seven linear codes over $\FF_4$ with generator matrices 

$G_1 = \left(\begin{array}{ccccccccccccccc}
1 & 1 & 1 &1 & 1 & 1 &1 & 1 & 1 & 1 &1 & 1 & 1 &1&1  
\end{array}\right)$, 

$G_2 = \left(\begin{array}{ccccccccccccccc}
1 & 1 & 1 &1 & 1 & 1 &1 & 1 & 1 & 1 &1 & 1 & 1 &1&1 \\
1 & \omega & \omega^2 & \omega^3 & \omega^4 & \omega^5 & \omega^6 & \omega^7 & \omega^8 & \omega^9 & \omega^{10} & \omega^{11} & \omega^{12} &\omega^{13}&\omega^{14}\\
1&\omega^{14}& \omega^{13} & \omega^{12} & \omega^{11} & \omega^{10} & \omega^9 & \omega^8& \omega^7 & \omega^6 & \omega^5 & \omega^4 & \omega^3 & \omega^2 & \omega 
\end{array}\right)$, 

$G_3 = \left(\begin{array}{ccccccccccccccc}
1 & 1 & 1 &1 & 1 & 1 &1 & 1 & 1 & 1 &1 & 1 & 1 &1&1 \\
1 & \omega & \omega^2 & \omega^3 & \omega^4 & \omega^5 & \omega^6 & \omega^7 & \omega^8 & \omega^9 & \omega^{10} & \omega^{11} & \omega^{12} &\omega^{13}&\omega^{14}\\
1&\omega^{14}& \omega^{13} & \omega^{12} & \omega^{11} & \omega^{10} & \omega^9 & \omega^8& \omega^7 & \omega^6 & \omega^5 & \omega^4 & \omega^3 & \omega^2 & \omega \\
1& \omega^2 & \omega^4 & \omega^6 & \omega^8 & \omega^{10} & \omega^{12} & \omega^{14} & \omega & \omega^3 & \omega^5 & \omega^7 & \omega^9 & \omega^{11} & \omega^{13} \\
1& \omega^{13} & \omega^{11} &\omega^{9} & \omega^{7} & \omega^5 & \omega^3 & \omega & \omega^{14} & \omega^{12} & \omega^{10} &\omega^8&\omega^6 & \omega^4 & \omega^2
\end{array}\right)$, 

$G_4 = \left(\begin{array}{ccccccccccccccc}
1 & 1 & 1 &1 & 1 & 1 &1 & 1 & 1 & 1 &1 & 1 & 1 &1&1 \\
1 & \omega & \omega^2 & \omega^3 & \omega^4 & \omega^5 & \omega^6 & \omega^7 & \omega^8 & \omega^9 & \omega^{10} & \omega^{11} & \omega^{12} &\omega^{13}&\omega^{14}\\
1&\omega^{14}& \omega^{13} & \omega^{12} & \omega^{11} & \omega^{10} & \omega^9 & \omega^8& \omega^7 & \omega^6 & \omega^5 & \omega^4 & \omega^3 & \omega^2 & \omega \\
1& \omega^2 & \omega^4 & \omega^6 & \omega^8 & \omega^{10} & \omega^{12} & \omega^{14} & \omega & \omega^3 & \omega^5 & \omega^7 & \omega^9 & \omega^{11} & \omega^{13} \\
1& \omega^{13} & \omega^{11} &\omega^{9} & \omega^{7} & \omega^5 & \omega^3 & \omega & \omega^{14} & \omega^{12} & \omega^{10} &\omega^8&\omega^6 & \omega^4 & \omega^2
\end{array}\right), \cdots$,
respectively. Now, by Construction \ref{con-1}, we define $\tilde{\CC}_1, \tilde{\CC}_2, \tilde{\CC}_3, \tilde{\CC}_4, \ldots$ with generator matrices

$\tilde G_1 = \left(\begin{array}{cccccccccccccccc}
\omega^{10} & \omega^{9} & 1 & \omega^5 & \omega^3 & \omega^8 & \omega & \omega^2& \omega^7 & \omega^{13} & \omega^4 & \omega^{14} & \omega^{12} & \omega^{11} & \omega^6 & 0\\
0&1 & 1 & 1 &1 & 1 & 1 &1 & 1 & 1 & 1 &1 & 1 & 1 &1&1  
\end{array}\right)$, 

$\tilde G_2 = \left(\begin{array}{cccccccccccccccc}
\omega^{10} & \omega^{9} & 1 & \omega^5 & \omega^3 & \omega^8 & \omega & \omega^2& \omega^7 & \omega^{13} & \omega^4 & \omega^{14} & \omega^{12} & \omega^{11} & \omega^6 & 0\\
0&1 & 1 & 1 &1 & 1 & 1 &1 & 1 & 1 & 1 &1 & 1 & 1 &1&1 \\
0&1 & \omega & \omega^2 & \omega^3 & \omega^4 & \omega^5 & \omega^6 & \omega^7 & \omega^8 & \omega^9 & \omega^{10} & \omega^{11} & \omega^{12} &\omega^{13}&\omega^{14}\\
0&1&\omega^{14}& \omega^{13} & \omega^{12} & \omega^{11} & \omega^{10} & \omega^9 & \omega^8& \omega^7 & \omega^6 & \omega^5 & \omega^4 & \omega^3 & \omega^2 & \omega 
\end{array}\right)$, 

$\tilde G_3 = \left(\begin{array}{cccccccccccccccc}
\omega^{10} & \omega^{9} & 1 & \omega^5 & \omega^3 & \omega^8 & \omega & \omega^2& \omega^7 & \omega^{13} & \omega^4 & \omega^{14} & \omega^{12} & \omega^{11} & \omega^6 & 0\\
0&1 & 1 & 1 &1 & 1 & 1 &1 & 1 & 1 & 1 &1 & 1 & 1 &1&1 \\
0&1 & \omega & \omega^2 & \omega^3 & \omega^4 & \omega^5 & \omega^6 & \omega^7 & \omega^8 & \omega^9 & \omega^{10} & \omega^{11} & \omega^{12} &\omega^{13}&\omega^{14}\\
0&1&\omega^{14}& \omega^{13} & \omega^{12} & \omega^{11} & \omega^{10} & \omega^9 & \omega^8& \omega^7 & \omega^6 & \omega^5 & \omega^4 & \omega^3 & \omega^2 & \omega \\
0&1& \omega^2 & \omega^4 & \omega^6 & \omega^8 & \omega^{10} & \omega^{12} & \omega^{14} & \omega & \omega^3 & \omega^5 & \omega^7 & \omega^9 & \omega^{11} & \omega^{13} \\
0&1& \omega^{13} & \omega^{11} &\omega^{9} & \omega^{7} & \omega^5 & \omega^3 & \omega & \omega^{14} & \omega^{12} & \omega^{10} &\omega^8&\omega^6 & \omega^4 & \omega^2
\end{array}\right)$, 

$\tilde G_4 = \left(\begin{array}{cccccccccccccccc}
\omega^{10} & \omega^{9} & 1 & \omega^5 & \omega^3 & \omega^8 & \omega & \omega^2 & \omega^7 & \omega^{13} & \omega^4 & \omega^{14} & \omega^{12} & \omega^{11} & \omega^6 & 0\\
0&1 & 1 & 1 &1 & 1 & 1 &1 & 1 & 1 & 1 &1 & 1 & 1 &1&1 \\
0&1 & \omega & \omega^2 & \omega^3 & \omega^4 & \omega^5 & \omega^6 & \omega^7 & \omega^8 & \omega^9 & \omega^{10} & \omega^{11} & \omega^{12} &\omega^{13}&\omega^{14}\\
0&1&\omega^{14}& \omega^{13} & \omega^{12} & \omega^{11} & \omega^{10} & \omega^9 & \omega^8& \omega^7 & \omega^6 & \omega^5 & \omega^4 & \omega^3 & \omega^2 & \omega \\
0&1& \omega^2 & \omega^4 & \omega^6 & \omega^8 & \omega^{10} & \omega^{12} & \omega^{14} & \omega & \omega^3 & \omega^5 & \omega^7 & \omega^9 & \omega^{11} & \omega^{13} \\
0&1& \omega^{13} & \omega^{11} &\omega^{9} & \omega^{7} & \omega^5 & \omega^3 & \omega & \omega^{14} & \omega^{12} & \omega^{10} &\omega^8&\omega^6 & \omega^4 & \omega^2
\end{array}\right), \cdots$,
respectively. It is easy to see that $\det\left(\tilde{G}_i\tilde{G}_i^\top\right)\neq \det\left(G_iG_i^\top\right)$, and $\det\left(\tilde{G}_i\tilde{G}_i^\top\right)=\omega^5$ for $1\leq i \leq 7$. Hence, $\tilde{\CC}_1=[16,2,15], \tilde{\CC}_2=[16,4,11], \tilde{\CC}_3=[16,6,9], \tilde{\CC}_4=[16,8,7], \cdots$ are equivalent to linear codes  with one-dimensional hull.
\end{example}
\begin{example}
Let $\omega$ is a primitive root of $\FF_{2^t}$ and $\CC_s:=[2^t-1,2s+1,2^t-2s-1]$ be linear code with generator matrix
$G_s = \left(\begin{array}{ccccc}
1 & 1 & 1 &\cdots&1 \\
1 & \omega & \omega^2 & &\omega^{2^t-2}\\
1 & \omega^{-1} & \omega^{-2} & \cdots &\omega^{-(2^t-2)}\\
\vdots & \vdots & \vdots & \ddots & \vdots \\
1 & \omega^{s} & \omega^{2s} & \cdots &\omega^{(2^t-2)s}\\
1 & \omega^{-s} & \omega^{-2s} & \cdots &\omega^{-(2^t-2)s}
\end{array}\right)$, where $0 \leq s \leq 2^{t-1}-1$. It can be checked that $\CC_0 = \{0\}$. Note that $\CC_s$ is an LCD, as $\det(G_sG_s^\top)=1$. Now, by Construction~\ref{con-1}, we define $\tilde{\CC_s}$ with generator matrix
$\tilde{G_s} = \left(\begin{array}{cccc}
\alpha & \textbf{P} \\
0 & G 
\end{array}\right)$ such that $\det\left(\tilde{G_s}\tilde{G_s}^\top\right)\neq \det\left(G_sG_s^\top\right)$,  where $\textbf{P}=(\beta_1~\beta_2~\cdots\beta_n)$ for $\beta_i\in\FF_q$ and $\alpha\in\FF_q^*$. For suitable choice of $\alpha$ and $\textbf{P}$, we obtain a class of codes $\tilde{\CC_s}:=[2^t,2s+2]$ such that $\tilde{\CC_s}$ is equivalent to a linear code  with one-dimensional hull.
\end{example}

\section{Some Constructions  of $\ell$-Dimensional Hull Codes}\label{sec:11dhc}

This section presents several methods for constructing a linear code's $\ell$-dimensional hull. For two linear codes $\CC_1:=[n,k_1]$ and $\CC_2:=[n,k_2]$ over the field $\FF_q$, we define their sum as $\CC_1+\CC_2 = \{c_1+c_2 \mid c_1\in\CC_1 \text{ and } c_2\in\CC_2\}$. It is straightforward to observe that $\CC_1+\CC_2$ forms a linear code of size $[n,k_1+k_2 - \dim(\CC_1\cap\CC_2)]$ over $\FF_q$.
\begin{thm}
Let $\CC_i:=[n,k_i]$ be two linear codes over $\FF_q$ with $\dim(\hull(\CC_i))=\ell_i$ for $i = 1, 2$. 
Let $\hull(\CC_1)\subseteq \CC_2^\perp$ and $\hull(\CC_2)\subseteq \CC_1^\perp$. If $\dim(\hull(\CC_1)\cap \hull(\CC_2))=\ell$, then $\dim(\hull(\CC_1+\CC_2))=\ell_1+\ell_2-\ell$. 
\end{thm}

\begin{proof}
Denote $\CC := \CC_1+\CC_2$. It can be checked that $(\CC_1+\CC_1)^\perp = \CC_1^\perp\cap\CC_2^\perp$
Then 
\begin{align*}
  \CC\cap\CC^{\perp}&=(\CC_1+\CC_2)\cap (\CC_1+\CC_1)^\perp &\\
                    &=(\CC_1+\CC_2)\cap (\CC_1^\perp\cap\CC_2^\perp)& \\
                    &=\left(\CC_1\cap (\CC_1^\perp\cap\CC_2^\perp)\right)+ \left(\CC_2\cap (\CC_1^\perp\cap\CC_2^\perp)\right) & \\
                    &=\left((\CC_1\cap\CC_1^\perp)\cap \CC_2^\perp\right) + \left((\CC_2\cap\CC_2^\perp)\cap \CC_1^\perp\right) &\\
                    &=\left(\CC_1\cap\CC_1^\perp\right) + \left(\CC_2\cap\CC_2^\perp\right)  ~~~~~~~~(\text{as }\hull(\CC_1)\subseteq \CC_2^\perp \text{ and } \hull(\CC_2)\subseteq \CC_1^\perp )&\\
%\end{align*}
%Next, \begin{align*}
\implies \dim\left(\CC\cap\CC^{\perp}\right)
       &=\dim\left(\CC_1\cap\CC_1^\perp\right) + \dim\left(\CC_2^\perp\cap\CC_2^\perp\right)-\dim\left(\left(\CC_1\cap\CC_1^\perp\right)\cap\left(\CC_2^\perp\cap\CC_2^\perp\right)\right)& \\
       &=\dim(\hull(\CC_1))+\dim(\hull(\CC_2))-\dim(\hull(\CC_1)\cap\hull(\CC_2))=\ell_1+\ell_2-\ell.&
\end{align*}
\end{proof}
\begin{cor}\label{cor-qq}
Let $\CC_i:=[n,k_i]$ be a linear codes over $\FF_q$ for $i\in\{1,2\}$ and $\CC=\CC_1+\CC_2$. Then, the following assertions hold.
\begin{enumerate}
\item[a)] If $\CC_1$ and $\CC_2$ are LCD codes, then $\CC$ is an LCD code.
\item[b)] If $\dim(\hull(\CC_1)) = \dim(\hull(\CC_2)) = 1$, then $\dim(\hull(\CC))\in \{0, 1, 2\}$.\\
Furthermore, if $\hull(\CC_1)\subseteq \CC_2^\perp$ and $\hull(\CC_2)\subseteq \CC_1^\perp$ then $\dim(\hull(\CC))\in \{1, 2\}$.
\item[c)] If $\dim(\hull(\CC_i))=1$ and $\dim(\hull(\CC_j))=0$ for $i\neq j$, then $\dim(\hull(\CC))\in \{0, 1\}$.\\
Further, if $\hull(\CC_i)\subseteq \CC_j^\perp$ then $\dim(\hull(\CC)) = 1$.
\item[d)] If $\dim(\hull(\CC_i))=\ell$ and $\dim(\hull(\CC_j))=0$ for $i\neq j$, then $\dim(\hull(\CC))\leq \ell$.
\end{enumerate}
\end{cor}

Next, we observe the following result.
\begin{thm}\label{th-as}
Let $\CC_i:=[n,k_i]$ be a linear code over $\FF_q$ for $i\in\{1,2\}$, and let $\CC=\CC_1+\CC_2$. If $\dim(\hull(\CC_1))=\ell$ and $\CC_2\subseteq \CC_1+\CC_1^\perp$, then $\dim(\hull(\CC))\geq \ell$. Moreover, if $\dim(\hull(\CC_2))=0$ (i.e., $\CC_2$ is LCD), then $\dim(\hull(\CC))= \ell$.
\end{thm}

\begin{proof}
Let $G_i$ be a generator matrix of $\CC_i$ for $i=1,2$. Also let $\{x_1,x_2,x_3,\ldots,x_\ell\}$ is a basis of $\hull(\CC_1)$. Then $x_i \in \CC$ for $1 \leq i \leq \ell$.
It is easy to observe that $G=\left(\begin{array}{cc}
G_1 \\ G_2
\end{array}\right)$ is a generating set of $\CC$.

Since $x_i\in\hull(\CC_1)$, $c x_i^\top = 0$ and $d x_i^\top = 0$ for all $c\in \CC_1$ and $d \in \CC^\perp_1$. It follows that $G_1x_i^\top=0$.
By hypothesis, we have $\CC_2\subseteq \CC_1+\CC_1^\perp$; it implies that the basis of $\CC_2$ is an element of a linear combination of the basis set of $\CC_1+\CC_1^\perp$, i.e., $G_2x_i^\top=0$. Hence, $Gx_i^\top = 0$ and that implies $x_i \in \CC^\perp$.
Therefore, $x_i \in \CC \cup \CC^\perp$ for $1 \leq i \leq \ell$ and hence $\dim(\hull(\CC))\geq \ell$.
The final part follows from Corollary~\ref{cor-qq}.
\end{proof}
The following result presents a canonical construction of $\ell+1$-dimensional hull from $\ell$-dimensional hull. This result was obtained in \cite{LSK23} for some special codes over $\FF_2$.
\begin{prop}\label{p-qwq}
Let $\CC:=[n,k]$ be a linear code over $\FF_q$ with $\dim(\hull(\CC))=\ell$ and a generator matrix $G$. If $d$ is a non-zero codeword of $\CC^\perp$ such that $d d^\top = 1$, then the extended code $\CC_{ex}:=[n+1,k+1]$ of $\CC$ with following generator matrix 
$G_{ex}=\left(\begin{array}{cc}
 1 & d \\
 0 & G
\end{array}\right)$ generates $\ell+1$-dimensional hull.
\end{prop}
\begin{proof}
It is easy to see that $\CC_{ex}$ is an $[n+1,k+1]$-linear code.
Next,
$G_{ex}G_{ex}^\top=\left(\begin{array}{cc}
 1 + d d^\top & d G^\top \\
 0 & G G^\top
\end{array}\right)=\left(\begin{array}{cc}
 0 & 0 \\
 0 & G G^\top
\end{array}\right)$.
Hence, $rank(G_{ex}G_{ex}^\top)=rank(GG^\top)$. 
As $\dim(\hull(\CC))=\ell$, by Proposition~\ref{pr-2.1}, we obtain that $rank(GG^\top) = k-\dim(\hull(\CC)) = k-\ell$.
Therefore from Proposition~\ref{pr-2.1} $\dim(\hull(\CC_{ex})) = k+1-rank(G_{ex}G_{ex}^\top) = k+1-rank(GG^\top) = k+1-(k-\ell) = \ell+1$, which completes the proof.
\end{proof}
Now, we present some numerical examples of construction.

\begin{example} Let \(\CC_1\) and \(\CC_2\) be two linear codes over the field \(\FF_4 = \{0, 1, \omega, \omega^2\}\) with the following generator matrices:
\[
G_1 = \begin{pmatrix} 
1 & 0 & \omega & \omega^2 \\ 
0 & 1 & \omega^2 & 1 
\end{pmatrix}
\]

and 

\[
G_2 = \begin{pmatrix} 
1 & \omega^2 & 1 & 1 
\end{pmatrix}.
\]

It is straightforward to verify that both \(\CC_1\) and \(\CC_2\) are linear codes that are also LCD (Linear Complementary Dual) codes. We can confirm using the software Magma that the sum \(\CC_1 + \CC_2\) results in a \([4,3,2]\) LCD code, which aligns with the result stated in Corollary~\ref{cor-qq}.
\end{example}

\begin{example}

Let \(\CC_1\) and \(\CC_2\) be two linear codes over the field \(\FF_4 = \{0, 1, \omega, \omega^2\}\) with the following generator matrices:

\[
G_1 = \begin{pmatrix}
1 & 0 & 0 & \omega & \omega^2 \\
0 & 1 & 0 & \omega & 1 
\end{pmatrix}
\]

and 

\[
G_2 = \begin{pmatrix}
0 & 0 & 1 & \omega^2 & \omega 
\end{pmatrix}.
\]

It is straightforward to verify that \(\CC_1\) is a linear code with the property that it is an LCD code, while \(\CC_2\) is a one-dimensional hull code. Using the software Magma, we can confirm that the sum \(\CC_1 + \CC_2\) forms a \([5, 3, 2]\) one-dimensional hull code, which aligns with Theorem~\ref{th-as}.
\end{example}

\begin{example}
Let $\CC:=[5,2,4]$ be a linear code over $\FF_8 :=\{0,1,\omega,\omega^2,\omega^3,\omega^4,\omega^5,\omega^6~|~1+\omega+\omega^3=0\}$ with generator matrices 
$G=\left(\begin{array}{ccccc}
1 & 0 & \omega & \omega^4 & \omega^6\\
0 & 1 & \omega^4 & \omega & \omega^5 \\
 \end{array}\right)$. It is straightforward to see that $\CC$ is an LCD and $(0,~0,~1,~\omega^5,~\omega^5)\in\CC^\perp$. An extended code $\CC_{ex}$ of $\CC$ has a generator matrix $G_{ex}=\left(\begin{array}{cccccc}
1 & 0 & 0 & 1     & \omega^5 & \omega^5 \\
0 & 1 & 0 & \omega & \omega^4 & \omega^6\\
0 & 0 & 1 & \omega^4 & \omega & \omega^5 \\
 \end{array}\right)$. Using Magma, we can verify that $\CC_{ex}$ is a $[6,3,3]$ one-dimensional hull, which coincides with Proposition~\ref{p-qwq}.
\end{example}

\begin{example}
Let $\CC:=[6,4,3]$ be a linear code over $\FF_8 :=\{0,1,\omega,\omega^2,\omega^3,\omega^4,\omega^5,\omega^6~|~1+\omega+\omega^3=0\}$ with generator matrices 
$G=\left(\begin{array}{cccccc}
1 & 0 & 0 & 0 & \omega^3 & \omega^2\\
0 & 1 & 0 & 0 & 1 & \omega \\
0 & 0 & 1 & 0 & \omega^6 & \omega^2 \\
0 & 0 & 0 & 1 & 1 & \omega^2 \\
 \end{array}\right)$. It is clear that $\CC$ is a one-dimensional hull and $(1,~0,~\omega^5,~\omega^3,~1,~\omega^6)\in\CC^\perp$. An extended code $\CC_{ex}$ of $\CC$ with a generator matrix $G=\left(\begin{array}{ccccccc}
1 & 1 & 0 & \omega^5 & \omega^3 & 1 & \omega^6\\
0 & 1 & 0 & 0 & 0 & \omega^3 & \omega^2\\
0 & 0 & 1 & 0 & 0 & 1 & \omega \\
0 & 0 & 0 & 1 & 0 & \omega^6 & \omega^2 \\
0 & 0 & 0 & 0 & 1 & 1 & \omega^2 \\
 \end{array}\right)$. Using Magma, we can verify that $\CC_{ex}$ is a $[7,5,3]$ two-dimensional hull, which coincides with Proposition~\ref{p-qwq}. Furthermore, $\CC$ is an MDS code.
\end{example}
\section{Conclusions}\label{sec: conclusion}

This paper focused on and advanced the construction of one-dimensional hull linear codes over binary finite fields. It also utilises zero-dimensional hull linear codes, which are known as LCD codes. 

Firstly, we highlight significant recent results (2023) by Chen, who demonstrated that any LCD code over an extended binary field \(\FF_q\) (with \(q > 3\)) that has a minimum distance of at least 2 is equivalent to the one-dimensional hull of a linear code. %This equivalence is established under a weaker condition than that originally presented by Chen.

Finally, the techniques employed in this study may lead to the development of other low-dimensional hull linear codes over finite fields, suggesting interesting future research directions. Our work contributes to constructing hulls with an \(\ell + 1\)-dimensional structure from hulls with an \(\ell\)-dimensional structure of linear codes.

\section*{Acknowledgments}
The authors thank the Associate Editor and the anonymous referees for their helpful comments and valuable suggestions, which have greatly improved the quality of the paper.

\end{document}